\pdfoutput=1
\RequirePackage{ifpdf}
\ifpdf % We are running pdfTeX in pdf mode
\documentclass[pdftex]{sigma}
\else
\documentclass{sigma}
\fi

\usepackage{bm}
\usepackage[mathscr]{euscript}

\numberwithin{equation}{section}

\newtheorem{thm}{Theorem}[section]
\newtheorem{prop}[thm]{Proposition}
\newtheorem{lem}[thm]{Lemma}
\newtheorem{cor}[thm]{Corollary}
\newtheorem{conj}[thm]{Conjecture}

\theoremstyle{definition}
\newtheorem*{eg}{Example}

\theoremstyle{remark}

\newcommand{\C}{\mathbb{C}}

\newcommand{\Z}{\mathbb{Z}}

\newcommand{\mc}{\mathcal}

\newcommand{\g}{\mathfrak{g}}
\newcommand{\gl}{\mathfrak{gl}}
\newcommand{\h}{\mathfrak{h}}
\newcommand{\n}{\mathfrak{n}}

\newcommand{\zg}{\mathfrak{z}(\widehat{\g})}

\newcommand{\End}{\mathrm{End}}

\newcommand{\sing}{{\mathrm{sing}}}

\newcommand{\pa}{\partial}

\newcommand{\gge}{\geqslant}
\newcommand{\lle}{\leqslant}
\newcommand{\la}{\lambda}

\newcommand{\bla}{\bm\lambda}

\newcommand{\bmx}{\begin{pmatrix}}
\newcommand{\emx}{\end{pmatrix}}

\newcommand{\Ug}{\mathrm{U}(\mathfrak{g})}

\begin{document}

\allowdisplaybreaks

\newcommand{\arXivNumber}{2008.06825}

\renewcommand{\thefootnote}{}

\renewcommand{\PaperNumber}{132}

\FirstPageHeading

\ShortArticleName{Perfect Integrability and Gaudin Models}

\ArticleName{Perfect Integrability and Gaudin Models\footnote{This paper is a~contribution to the Special Issue on Representation Theory and Integrable Systems in honor of Vitaly Tarasov on the 60th birthday and Alexander Varchenko on the 70th birthday. The full collection is available at \href{https://www.emis.de/journals/SIGMA/Tarasov-Varchenko.html}{https://www.emis.de/journals/SIGMA/Tarasov-Varchenko.html}}}

\Author{Kang LU~}

\AuthorNameForHeading{K.~Lu}

\Address{Department of Mathematics, University of Denver, 2390 S.~York St., Denver, CO 80210, USA}
\Email{\href{mailto:Kang.Lu@du.edu}{Kang.Lu@du.edu}}
\URLaddress{\url{https://kanglu.me}}

\ArticleDates{Received August 26, 2020, in final form December 02, 2020; Published online December 10, 2020}

\Abstract{We suggest the notion of perfect integrability for quantum spin chains and conjecture that quantum spin chains are perfectly integrable. We show the perfect integrability for Gaudin models associated to simple Lie algebras of all finite types, with periodic and regular quasi-periodic boundary conditions.}


\Classification{82B23; 17B80}

\renewcommand{\thefootnote}{\arabic{footnote}}
\setcounter{footnote}{0}

\section{Introduction}
Quantum spin chains are important models in integrable system. These models have numerous deep connections with other areas of mathematics and physics. In this article, we would like to suggest the notion of perfect integrability for quantum spin chains.

Let us recall Gaudin models and XXX spin chains. Let $\g$ be a simple (or reductive) Lie (super)algebra and $G$ the corresponding Lie group. Let $\mathscr A_{\g}$ be an affinization of $\g$ where the universal enveloping algebra $\Ug$ of $\g$ can be identified as a Hopf subalgebra of $\mathscr A_{\g}$. Here $\mathscr A_{\g}$ is either the universal enveloping algebra of the current algebra $\mathrm{U}(\g[t])$ which describes the symmetry for Gaudin models, or Yangian $\mathrm{Y}(\g)$ associated to $\g$ for XXX spin chains. In both cases the algebra $\mathscr A_{\g}$ has a remarkable commutative subalgebra called the \emph{Bethe algebra}. We denote the Bethe algebra by $\mathscr B_{\g}$. The Bethe algebra $\mathscr B_{\g}$ commutes with $\Ug$. Take any finite-dimensional irreducible representation $M$ of $\mathscr A_{\g}$, then $\mathscr B_{\g}$ acts naturally on the space of singular vectors $M^\sing$. Let $\mathscr B_{\g}(M^\sing)$ be the image of $\mathscr B_{\g}$ in $\End(M^\sing)$. The problem is to study the spectrum of $\mathscr B_{\g}(M^\sing)$ acting on $M^\sing$.\footnote{The reason these models are called spin chains is that $M$ is usually a tensor product of evaluation modules where each factor corresponds to a particle of some spin.}

With the agreement with the philosophy of geometric Langlands correspondence, it is important to understand and describe the finite-dimensional algebra $\mathscr B_{\g}(M^\sing)$ and the corresponding scheme $\mathsf{Spec}(\mathscr B_{\g}(M^\sing))$. Or more generally, find a geometric object parameterizing the eigenspaces of $\mathscr B_{\g}$ when $M$ runs over all finite-dimensional irreducible representations (up to isomorphism). In Gaudin models, the underlying geometric objects are described by the sets of monodromy-free $^L\g$-opers with regular singularities of prescribed residues at evaluation points, see \cite{FFR10,Ryb18}, where $^L\g$ is the Langlands dual of $\g$. Moreover, when $\g=\gl_{N}$, the algebra $\mathscr B_{\g}(M^\sing)$ is interpreted as the space of functions on the intersection of suitable Schubert cycles in a Grassmannian variety, see~\cite{MTV09}. This interpretation gives a relation between representation theory and Schubert calculus useful in both directions which has important applications in real algebraic geometry, see \cite{MT16, MTV09}.

Any finite-dimensional unital commutative algebra $\mc B$ is a module over itself induced by left multiplication. We call this module the \emph{regular representation of} $\mc B$. The dual space $\mc B^*$ is naturally a $\mc B$-module which is called the \emph{coregular representation}. A Frobenius algebra is a~finite-dimensional unital commutative algebra whose regular and coregular representations are isomorphic, see Section~\ref{sec frob-alg}.

Let $V$ be a finite-dimensional $\mathcal B$-module. Let $\mc B(V)$ be the image of $\mc B$ in $\End(V)$. We say that the $\mathcal B$-module $V$ is \emph{perfectly integrable} if $\mathcal B$ acts on $V$ cyclically and the algebra $\mc B(V)$ is a~Frobenius algebra. Note that in this case, the $\mc B(V)$-module $V$ is isomorphic to the regular and coregular representations of $\mc B(V)$.

Based on the extensive study of quantum spin chains, see the evidences from \cite{CLV20,FFR10,LM19,MTV08,MTV09,MTV14,Ryb18}, the following conjecture is expected to hold.

\begin{conj}\label{conj frob-int}
The $\mathscr B_{\g}$-module $M^\sing$ is perfectly integrable.
\end{conj}

\looseness=1 In fact there is a family of commutative Bethe algebras $\mathscr B_{\g}^\mu$ depending on an element $\mu\in\g^*$ (resp.~$\mu\in G$). From here to Conjecture~\ref{conj frob-int quasi}, we use the parenthesis to indicate the modifications for XXX spin chains. If $\mu\in\g^*$ (resp.~$\mu\in G$) is a regular semi-simple element, we say that the corresponding spin chain has regular quasi-periodic boundary condition. Moreover, if $\mu=0$ (resp. $\mu=\mathrm{Id}$), then the algebra $\mathscr B_{\g}^0$ (resp.~$\mathscr B_{\g}^{\mathrm{Id}}$) coincides with the algebra $\mathscr B_{\g}$ considered above. If $\mu=0$ (resp.~$\mu=\mathrm{Id}$), we say that the corresponding spin chain has periodic boundary condition.

For regular quasi-periodic spin chains the Bethe algebra does not commute with $\mathrm U(\g)$ and one replaces $M^\sing$ with $M$. Denote by $\mathscr B_{\g}^\mu(M)$ the image of $\mathscr B_{\g}^\mu$ in $\End(M)$. Let $\h$ be the Cartan subalgebra of $\g$ and $H$ the Cartan subgroup of $G$.

\begin{conj}\label{conj frob-int quasi}
If $\mu\in \h^*$ $($resp.~$\mu\in H)$ is regular, then the $\mathscr B_{\g}^\mu$-module $M$ is perfectly integrable.
\end{conj}

For more general $\mu\in\g^*$, one has to replace $M^\sing$ or $M$ with an appropriate subspace of $M$ depending on $\mu$, see Conjecture~\ref{conj frob-ger} in Section~\ref{sec general conj}.

When Conjectures \ref{conj frob-int}, \ref{conj frob-int quasi}, and \ref{conj frob-ger} hold, we say that the corresponding quantum spin chains are \emph{perfectly integrable}.

The perfect integrability was shown for
\begin{itemize}\itemsep=0pt
 \item Gaudin models of $\gl_N$ in \cite{MTV08,MTV09} with periodic and regular quasi-periodic boundary conditions;
 \item XXX (resp.~XXZ) spin chains of $\gl_N$ associated to irreducible tensor products of vector representations in~\cite{MTV14} (resp.~\cite{RTV14}) with periodic and regular quasi-periodic boundary conditions;
 \item XXX spin chains of $\gl_{1|1}$ associated to cyclic tensor products of polynomial representations in~\cite{LM19} with periodic and regular quasi-periodic boundary conditions;
 \item XXX spin chains of $\gl_{m|n}$ associated to irreducible tensor products of vector representations in~\cite{CLV20} with periodic boundary condition.
\end{itemize}

Our main result confirms Conjectures~\ref{conj frob-int} and~\ref{conj frob-int quasi} for Gaudin models of all finite types, see Theorem \ref{thm main}. We deduce Theorem \ref{thm main} from
\cite[Corollary~5]{FFR10}, \cite[Theorem~3.2]{Ryb18}, and \cite[Theorem~8.1.5]{Fre07}.

Our suggestion to call the situations in Conjectures~\ref{conj frob-int} and~\ref{conj frob-int quasi} ``perfect integrability" is motivated by Lemma~\ref{cor intro} below.

Let $\mathcal B$ be a finite-dimensional unital commutative algebra. Let $V$ be a finite-dimensional $\mathcal B$-module and $\mathcal E\colon \mathcal B\to \C$ a character, then the {\it $\mc B$-eigenspace} and {\it generalized $\mc B$-eigenspace} associated to $\mc E$ in $V$ are defined by{\samepage
\[
\bigcap_{a\in \mathcal B}\ker(a|_V-\mathcal E(a))\qquad\text{and}\qquad \bigcap_{a\in \mathcal B}\big(\bigcup_{m=1}^\infty\ker(a|_V-\mathcal E(a))^m\big),
\]
respectively. Let $\mc B(V)$ be the image of $\mc B$ in $\End(V)$.}

\begin{lem}\label{cor intro}
If the $\mc B$-module $V$ is perfectly integrable, then every $\mc B$-eigenspace in $V$ has dimension one, and there exists a bijection between $\mc B$-eigenspaces in $V$ and $\mathsf{Specm}(\mc B(V))$~-- the subset of closed points in $\mathsf{Spec}(\mc B(V))$. Moreover, each generalized $\mc B$-eigenspace is a cyclic $\mc B$-module, and the algebra $\mc B(V)$ is a maximal commutative subalgebra in $\End(V)$ of dimension $\dim V$.
\end{lem}
This lemma easily follows from general well-known facts about regular and coregular representations of a finite-dimensional unital commutative algebra, see, e.g., \cite[Section~3.3]{MTV09}. We provide a proof of Lemma~\ref{cor intro} in Section~\ref{sec frob-alg}.

Note that we expect that the dimensions of eigenspaces are one from the general philosophy of Bethe ansatz conjecture. The integrability in any sense always asserts that the algebra of Hamiltonians is maximal commutative. And we also expect that the Bethe algebra has geometric nature based on the geometric Langlands correspondence~\cite{Fre07}.

In the case of Gaudin models, it is proved in \cite[Theorem~3.2]{Ryb18} (resp.\ \cite[Corollary~5]{FFR10}) that $\mathscr B_{\g}$ (resp.~$\mathscr B^\mu_{\g}$ with regular $\mu$) acts cyclically on $M^{\sing}$ (resp.~$M$). For generic values of evaluation parameters (in the periodic case or in the case of generic regular $\mu\in \h^*$) the action of Bethe algebra is diagonalizable and we immediately obtain that eigenspaces have dimension one. However, we cannot make such a conclusion for {\it arbitrary} parameters. Indeed, if a linear operator acts cyclically on a vector space then all its eigenspaces have dimension one. But the same result fails if we replace a single operator by a set of commuting linear operators, as the following simple example shows.

\begin{eg}\label{eg counter}Let $\mathscr A=\C[x_1,x_2]/\big\langle x_1^2,x_2^2,x_1x_2\big\rangle$. Consider the regular representation~$\mathscr A$. Then the eigenspace corresponding to the trivial character is spanned by~$x_1$ and~$x_2$ which is two-dimensional.
\end{eg}

We supplement the results of \cite{FFR10} and \cite{Ryb18} with the nondegenerate symmetric bilinear form on~$M$ which makes $\mathscr B_{\g}^\mu(M)$ Frobenius which allows us to use Lemma~\ref{cor intro}. The bilinear form comes from the tensor product of Shapovalov forms on $M$, we show that all elements of Bethe algebra $\mathscr B_\g^\mu(M)$ with $\mu\in\h^*$ are symmetric with respect to this form, see Lemma~\ref{lem inv-form}.

In the rest of the paper, we only deal with Gaudin models. We refer the readers to \cite{IR20} for details about the Bethe algebra of Yangian $\mathrm{Y}(\g)$ (XXX spin chains). We expect the conjectures with proper modifications also hold for XXZ and XYZ spin chains.

\section{Perfect integrability of Gaudin models}\label{sec pre}
\subsection{Feigin--Frenkel center}\label{sec bethe}
In this section, we recall the definition of Feigin--Frenkel center and its properties.

Let $\g$ be a complex simple Lie algebra of rank $r$. Consider the affine Kac--Moody algebra $\widehat \g$,\[\widehat \g=\g\big[t,t^{-1}\big]\oplus \C K,\qquad \g\big[t,t^{-1}\big]=\g\otimes \C\big[t,t^{-1}\big],\]where $\C\big[t,t^{-1}\big]$ is the algebra of Laurent polynomials in~$t$. For $X\in\g$ and $s\in\Z$, we simply write~$X[s]$ for~$X\otimes t^s$. Let $\g_-=\g\otimes t^{-1}\C\big[t^{-1}\big]$ and $\g[t]=\g\otimes \C[t]$.

Let $h^\vee$ be the \emph{dual Coxeter number} of $\g$. Define the module $V_{-h^\vee}(\g)$ as the quotient of $\mathrm{U}(\widehat \g)$ by the left ideal generated by $\g[t]$ and $K+h^\vee$. We call the module $V_{-h^\vee}(\g)$ the \emph{Vaccum module at the critical level over} $\widehat \g$. The vacuum module $V_{-h^\vee}(\g)$ has a vertex algebra structure.

Define the subspace $\mathfrak z(\widehat\g)$ of $V_{-h^\vee}(\g)$ by
\[
\mathfrak z(\widehat \g)=\{v\in V_{-h^\vee}(\g)\,|\, \g[t]v=0\}.
\]
Using the PBW theorem, it is clear that $V_{-h^\vee}(\g)$ is isomorphic to $\mathrm U(\g_-)$ as vector spaces. There is an injective homomorphism from $\mathfrak z(\widehat\g)$ to $\mathrm{U}(\g_-)$. Hence $\mathfrak z(\widehat\g)$ is identified as a commutative subalgebra of $\mathrm{U}(\g_-)$. We call $\mathfrak z(\widehat\g)$ the \emph{Feigin--Frenkel center}.

There is a distinguished element $S_1\in \mathfrak z(\widehat \g)$ given by
\[
S_1=\sum_{a=1}^{\dim \g}X_a[-1]^2,
\]
where $\{X_a\}$ is an orthonormal basis of $\g$ with respect to the Killing form. The element $S_1$ is called the \emph{Segal--Sugawara vector}.

\begin{prop}[\cite{Ryb08}]\label{prop central}
The subalgebra $\mathfrak z(\widehat\g)$ is the centralizer of $S_1$ in $\mathrm{U}(\g_-)$.
\end{prop}

Let $e_{1},\dots,e_{r}$, $h_1,\dots,h_r$, $f_1,\dots,f_r$ be a set of Chevalley generators of $\g$. Let $\varpi\colon \g\to \g$ be the \textit{Cartan anti-involution} sending $e_{1},\dots,e_{r}$, $h_1,\dots,h_r$, $f_1,\dots,f_r$ to $f_{1},\dots,f_{r}$, $h_1,\dots,h_r$, $e_1,\dots,e_r$, respectively. Let $\widehat \varpi$ be the anti-involution on $\widehat\g$ defined by
\[
\widehat\varpi\colon \ \widehat\g\to \widehat \g,\qquad X[s]\mapsto \varpi(X)[s],
\]
for all $X\in \g$ and $s\in \Z$. We also call $\widehat\varpi$ the Cartan anti-involution.

\begin{cor}\label{cor invariant}
The Feigin--Frenkel center $\mathfrak z(\widehat\g )$ is invariant under $\widehat \varpi$.
\end{cor}
\begin{proof}
Since by Proposition \ref{prop central}, $\mathfrak z(\widehat\g )$ is the centralizer of $S_1$ in $\mathrm{U}(\g_-)$, the statement follows from the fact that $\widehat\varpi(S_1)=S_1$.
\end{proof}

\subsection{Affine Harish-Chandra homomorphism}
Let $\n_+$ be the nilpotent Lie subalgebra generated by $e_{1}, \dots, e_{r}$. Let $\n_-$ be the nilpotent Lie subalgebra generated by $f_{1},\dots,f_{r}$. Let~$\h$ be the Cartan subalgebra generated by $h_{1},\dots,h_{r}$. One has the triangular decomposition $\g=\n_-\oplus \h\oplus \n_+$.

The Lie algebra $\g$ is considered as a subalgebra of $\widehat\g$ via identifying $X\in\g$ with $X[0]\in \widehat \g$. The Lie subalgebra $\h$ acts on $\widehat\g$ adjointly and hence acts on $\mathrm U(\g_-)$. Let $\mathrm U(\g_-)^\h$ be the centralizer of~$\h$ in~$\mathrm U(\g_-)$.

Let $J$ be the intersection of $\mathrm U(\g_-)^\h$ and the left ideal of $\mathrm U(\g_-)$ generated by $t^{-1}\n_-[t^{-1}]$. Then we have the direct sum of vector spaces,
\begin{gather}\label{dec centralizer}
\mathrm U(\g_-)^\h=\mathrm{U}\big(t^{-1}\h\big[t^{-1}\big]\big)\oplus J.
\end{gather}
Hence we have the projection
\[
\mathfrak f\colon \ \mathrm U(\g_-)^\h\to \mathrm{U}\big(t^{-1}\h\big[t^{-1}\big]\big).
\]It is clear that $\mathfrak f$ is a homomorphism of algebras. We call $\mathfrak f$ the {\it affine Harish-Chandra homomorphism}. We use the same letter $\mathfrak f$ for the restriction map $\mathfrak f\colon \mathfrak z(\widehat \g)\to \mathrm{U}\big(t^{-1}\h\big[t^{-1}\big]\big)$.

The following proposition is a part of \cite[Theorem~8.1.5]{Fre07}.
\begin{prop}\label{prop injectivity}
The homomorphism $\mathfrak f\colon \mathfrak z(\widehat \g)\to \mathrm{U}\big(t^{-1}\h\big[t^{-1}\big]\big)$ is injective.
\end{prop}

Using Proposition \ref{prop injectivity}, we improve Corollary~\ref{cor invariant} to the following proposition.
\begin{prop}\label{prop fix B}
For any element $S\in \mathfrak z(\widehat\g )$, we have $\widehat \varpi(S)=S$.
\end{prop}
The proposition was proved in \cite[Proposition~8.4]{MTV06} for type A and in \cite[Proposition~6.1]{Lu18} for types B and C.
\begin{proof}
Now take $S\in \zg$ and write the decomposition of $S$ as in \eqref{dec centralizer}, $S=S_{\h}+S_{\mathfrak j}$, where $S_{\h}\in \mathrm{U}\big(t^{-1}\h\big[t^{-1}\big]\big)$ and $S_{\mathfrak j}\in J$. Then $\widehat\varpi(S)=\widehat\varpi(S_{\h})+\widehat\varpi(S_{\mathfrak j})$. Note that $\widehat\varpi$ fix elements in $\mathrm{U}\big(t^{-1}\h\big[t^{-1}\big]\big)$ and $S_{\h}\in \mathrm{U}\big(t^{-1}\h\big[t^{-1}\big]\big)$ we have $\widehat\varpi(S_{\h})=S_{\h}$. Note also that $\widehat\varpi$ maps $\mathrm{U}\big(t^{-1}\n_+\big[t^{-1}\big]\big)$ to $\mathrm{U}\big(t^{-1}\n_-\big[t^{-1}\big]\big)$ and $\mathrm{U}\big(t^{-1}\n_-\big[t^{-1}\big]\big)$ to $\mathrm{U}\big(t^{-1}\n_+\big[t^{-1}\big]\big)$, we have $\widehat\varpi(S_{\mathfrak j})\in J$ since $J$ is the intersection of the $\h$-centralizer $\mathrm{U}(\g_-)^\h$ with the left ideal of $\mathrm U(\g_-)$ generated by $t^{-1}\n_-\big[t^{-1}\big]$ and also the intersection of $\mathrm U(\g_-)^\h$ with the right ideal of $\mathrm U(\g_-)$ generated by $t^{-1}\n_+\big[t^{-1}\big]$. It follows that
\[ \mathfrak f(S)=S_\h=\mathfrak f\circ \widehat\varpi(S).\]
Note that by Corollary \ref{cor invariant} both $S$ and $\widehat\varpi (S)$ are elements in $\zg$. Since by Proposition~\ref{prop injectivity} the homomorphism $\mathfrak f\colon \mathfrak z(\widehat \g)\to \mathrm{U}\big(t^{-1}\h\big[t^{-1}\big]\big)$ is injective, we conclude that $S=\widehat\varpi(S)$, completing the proof.
\end{proof}

\subsection{Gaudin models}
We recall the construction of Gaudin models from, e.g.,~\cite{Ryb06,Ryb18}. The coproduct of $\mathrm{U}(\g_-)$ is given by
\[
\Delta\colon \ X[s]\mapsto X[s]\otimes 1+ 1\otimes X[s], \qquad X\in \g,\qquad s <0.
\]
Let $\ell$ be a positive integer. Using the iterated coproduct, one has the homomorphism
\[
\mathrm{U}(\g_-)\to \mathrm{U}(\g_-)^{\otimes (\ell+1)}.
\]
For any $z\in\C^\times$, one gets the \emph{evaluation map} at $z$
\[
\varphi_{z}\colon \ \mathrm U(\g_-)\to \mathrm{U}(\g), \qquad X[s]\mapsto z^s X.
\]
For any $\mu\in\g^*$, one obtains the character
\[
\psi_{\mu}\colon \ \mathrm U(\g_-)\to \C,\qquad X[s]\mapsto \delta_{s,-1}\mu(X),
\]

Fix a sequence of pairwise distinct nonzero complex numbers $\bm z=(z_1,\dots,z_\ell)$. Then using these three homomorphisms, one obtains a new homomorphism
\[
\varphi_{\bm z,\mu}\colon \ \mathrm U(\g_-)\to \mathrm{U}(\g)^{\otimes \ell},\qquad \varphi_{\bm z,\mu}(X[s])=\sum_{a=1}^\ell z_a^s (X)_a+\delta_{s,-1}\mu(X),
\]
where $(X)_a=1^{\otimes (a-1)}\otimes X\otimes 1^{\otimes (\ell -a)}$.

Set $u-\bm{z}=(u-z_1,\dots,u-z_\ell)$. Define the \emph{Gaudin algebra} as a subalgebra generated by elements in $\varphi_{u-\bm z,\mu}(\mathfrak z(\widehat\g))\subset \mathrm{U}(\g)^{\otimes \ell}$ for all $u\in \C\setminus\{z_1,\dots,z_\ell\}$. The Gaudin algebra is commutative and it is denoted by $\mathcal A_{\bm z,\mu}$. When $\mu=0$, the Gaudin algebra commutes with the diagonal action of $\mathrm{U}(\g)$, see, e.g., \cite[Proposition~3]{Ryb06}.

Let $\bla=(\la_1,\dots,\la_\ell)$ be a sequence of dominant integral weights. Denote by $V_{\la_i}$ the finite-dimensional irreducible $\g$-module of highest weight $\la_i$. We set $V_{\bla}=\otimes_{i=1}^\ell V_{\la_i}$ and
\begin{gather*}%\label{eq M}
(V_{\bla})^\sing=\{v\in V_{\bla}\,|\, \n_+ v=0\},\qquad \mathcal M_{\bla,\mu}= \begin{cases}(V_{\bla})^\sing,& \text{if }\mu=0,\\ V_{\bla}, & \text{if }\mu\in\h^* \text{ is regular}.\end{cases}
\end{gather*}
Here we identify $\h^*$ with the subspace of $\g^*$ consisting of all elements annihilating $\n_+\oplus \n_-$. By the construction of $\mathcal A_{\bm z,\mu}$, $\mathcal M_{\bla,\mu}$ is an $\mathcal A_{\bm z,\mu}$-module.

Let $\mathcal A_{\bm z,\mu}$ be the algebra of Hamiltonians and $\mathcal M_{\bla,\mu}$ the spin chain. We call the corresponding integrable system the \emph{Gaudin model}. We say that the Gaudin model has \emph{periodic boundary condition} if $\mu=0$ and \emph{regular quasi-periodic boundary condition} if $\mu\in\h^*$ is regular. We would like to study the spectrum of $\mathcal A_{\bm z,\mu}$ acting on $\mathcal M_{\bla,\mu}$.

The following theorem is obtained in \cite[Corollary~5]{FFR10} for any regular $\mu\in\g^*$ and in \cite[Theorem~3.2]{Ryb18} for $\mu =0$.
\begin{thm}\label{thm cyclicity}
If $\mu\in \h^*$ is regular or if $\mu =0$, then the space $\mathcal M_{\bla,\mu}$ is cyclic as an $\mathcal A_{\bm z,\mu}$-module.
\end{thm}

\subsection{Bethe algebra}
Note that our definition of Gaudin models is slightly different from that in Introduction. In this section, we define the Bethe algebra in $\mathrm U(\g[t])$ and clarify this point.

Following, e.g., \cite[Section 5]{Mol12}, we recall the definition of Bethe algebra using Feigin--Frenkel center $\zg$. Note that the Feigin--Frenkel center $\zg$ is in $\mathrm U(\g_-)$ while the Bethe algebra is in $\mathrm U(\g[t])$.

For $\mathcal X\in \g$ and $\mu\in\g^*$, define the current $\mathcal X^{\mu}(u)$ by
\[
\mathcal X^{\mu}(u)=\mu(\mathcal X)+\sum_{r\gge 0}\mathcal X[r]u^{-r-1}\in \mathrm U(\g[t])\big[\big[u^{-1}\big]\big].
\]
For any element $a$ of the form
\[
a=\sum \mathcal X_{1}[-s_1-1]\mathcal X_{2}[-s_2-1]\cdots \mathcal X_{k}[-s_{k}-1]\in \zg
\]
for some $k\in\Z_{>0}$, $\mathcal X_{i}\in\g$, $s_i\in\Z_{\gge 0}$, define a series in $u^{-1}$ whose coefficients are in $\mathrm U(\g[t])$ by
\begin{gather}\label{eq:st}
\sum \frac{(-1)^k}{s_1!s_2!\cdots s_k!}\pa_u^{s_1}\mathcal X_1^{\mu}(u)\pa_u^{s_2}\mathcal X_2^{\mu}(u)\cdots \pa_u^{s_k}\mathcal X_k^{\mu}(u)\in \mathrm U(\g[t])\big[\big[u^{-1}\big]\big].
\end{gather}
The \emph{Bethe algebra} $\mathscr B^\mu_\g$ is the subalgebra of $\mathrm U(\g[t])$ generated by the coefficients of all such series of form~\eqref{eq:st} as~$a$ runs over~$\zg$.

The Bethe algebra $\mathscr B^\mu_\g$ (or Feigin--Frenkel center $\zg$) is considered as the universal Gaudin algebra, see, e.g.,~\cite{IR20}. The Gaudin algebra $\mathcal A_{\bm z,\mu}$ in $\mathrm U(\g)^{\otimes \ell}$ can also be obtained from $\mathscr B^\mu_\g$ by taking the $(\ell-1)$-th fold coproduct and then applying to the $i$-th factor the evaluation map at $z_i$ for every $1\lle i\lle \ell$. In particular, the image of the Gaudin algebra $\mathcal A_{\bm z,\mu}$ acting on $V_{\bla}$ coincides with that of Bethe algebra $\mathscr B_\g^\mu$ acting on tensor product of evaluation modules $V_{\bla}$ with evaluation points at $\bm z=(z_1,\dots,z_\ell)$, see \cite[Propositions 2.3 and 2.5]{Ryb18} or \cite{FFR94,FFT10,Ryb06}. Note that in this case, all finite-dimensional irreducible $\mathrm U(\g[t])$-modules are tensor products of evaluation modules with pairwise distinct evaluation parameters.

\subsection{Shapovalov form}
For a dominant integral weight $\la$, there is a unique nondegenerate symmetric bilinear form $S_{\la}$ on $V_{\la}$ such that
\[
\mc S_{\la}(v_\la,v_\la)=1,\qquad \mc S_{\la}(Xv,w)=\mc S_{\la}(v,\varpi(X)w),
\]
where $v_{\la}$ is a highest weight vector of $V_\la$ and $v,w\in V_\la$. We call $\mc S_\la$ the \emph{Shapovalov form on $V_\la$}. The Shapovalov form $\mc S_\la$ is positive definite on the real part of $V_\la$.

The Shapovalov forms $\mc S_{\la_i}$ induce a nondegenerate symmetric bilinear form $\mc S_{\bla}=\otimes_{i=1}^\ell \mc S_{\la_i}$ on~$V_{\bla}$. The restriction of $\mathcal S_{\bla}$ on the singular subspace $(V_{\bla})^\sing$ is also nondegenerate.

Suppose $\mu\in\h^*$, then it is clear that
\begin{gather}\label{eq inv-form}
\mathcal S_{\bla}(\varphi_{\bm z,\mu}(X[s])v,w)=\mathcal S_{\bla}(v,\varphi_{\bm z,\mu}(\varpi(X)[s])w)=\mathcal S_{\bla}(v,\varphi_{\bm z,\mu}\circ\widehat{\varpi}(X[s])w),
\end{gather}
for all $v,w\in V_{\bla}$ and $X\in\g$.

Let $\rho_{\bla,\bm z,\mu}\colon \mathcal A_{\bm z,\mu}\to \End(\mathcal M_{\bla,\mu})$ be the representation of the natural action of $\mathcal A_{\bm z,\mu}$ on $\mathcal M_{\bla,\mu}$. Let $\mathfrak A_{\bla,\bm z,\mu}$ be the image of $\mathcal A_{\bm z,\mu}$ under $\rho_{\bla,\bm z,\mu}$.
\begin{lem}\label{lem inv-form}
Let $a\in \mathfrak A_{\bla,\bm z,\mu}$ and $v,w\in \mathcal M_{\bla,\mu}$. If $\mu\in\h^*$, then we have $\mathcal S_{\bla}(av,w)=\mathcal S_{\bla}(v,aw)$.
\end{lem}
\begin{proof}
The statement follows from \eqref{eq inv-form} and Proposition \ref{prop fix B}.
\end{proof}

\subsection{Frobenius algebra}\label{sec frob-alg}
Let $A$ be a finite-dimensional commutative unital algebra. If there exists a nondegenerate symmetric bilinear form $(\cdot,\cdot)$ on $A$ such that
\[
(ab,c)=(a,bc)\qquad \text{for all } a,b,c\in A,
\]
then it is clear that the regular and coregular representations of~$A$ are isomorphic. Thus~$A$ is a~Frobenius algebra.

We prepare the following lemma for the proof of the main theorem. Suppose $A$ is a unital commutative algebra acting on a finite-dimensional space $V$, $\rho\colon A\to \End(V)$. Let $\mathfrak A$ be the image of $A$ under $\rho$ in $\End(V)$. Clearly, $\mathfrak A$ is a finite-dimensional unital commutative algebra.

\begin{lem}\label{lem frobenius}
Suppose $A$ acts on $V$ cyclically. If there is a nondegenerate symmetric bilinear form $(\cdot|\cdot)$ on $V$ such that \[ (av|w)=(v|aw),\qquad\text{for all } a\in \mathfrak A, \quad v,w\in V,\] then the algebra $\mathfrak A$ is a Frobenius algebra. In particular, the $A$-module $V$ is perfectly integrable.
\end{lem}
\begin{proof}
Let $v^+$ be a cyclic vector of the action of $\mathfrak A$ on $V$. Define a linear map $\xi$ by
\[ \xi\colon \ \mathfrak A\to V,\qquad a\mapsto av^+.\]
Clearly, $\xi$ is surjective.

We claim that $\xi$ is injective. Indeed, suppose that $a\in \ker\xi$, then $a\in\End(V)$ and $av^+=0$. Hence $aa'v^+=a'av^+=0$ for all $a'\in\mathfrak A$, namely $a\, \xi(\mathfrak A)=0$. Since $\xi(\mathfrak A)=V$, we conclude that $aV=0$. Therefore $a=0$, which implies $\xi$ is injective and hence a bijection. Then it is clear that~$\xi$ defines an $\mathfrak A$-module isomorphism between the regular representation of $\mathfrak A$ and the $\mathfrak A$-module~$V$.

Define a bilinear form $(\cdot,\cdot)$ on $\mathfrak A$ as follows,
\[
(a,b)=(av^+|bv^+),\qquad\text{for all } a,b\in \mathfrak A.
\]
Since $(\cdot|\cdot)$ is symmetric, so is $(\cdot,\cdot)$. Because $(\cdot|\cdot)$ is nondegenerate and $\xi$ is bijective, the form $(\cdot,\cdot)$ is nondegenerate as well. For $a,b,c\in\mathfrak A$, we also have
\[
(ab,c)=(abv^+|cv^+)=(bav^+|cv^+)=(av^+|bcv^+)=(a,bc).
\]
Hence $\mathfrak A$ is a Frobenius algebra.
\end{proof}

Since Lemma \ref{cor intro} is central to the results of the present paper, we also provide a detailed proof for it.

\begin{proof}[Proof of Lemma \ref{cor intro}]
	To simplify the notation, we write $\mathfrak B$ for $\mathcal B(V)$. If the $\mathcal B$-module $V$ is perfectly integrable, then the $\mathfrak B$-module $V$ is isomorphic to the regular representation $\mathfrak B$ and coregular representation $\mathfrak B^*$. Note that the $\mathcal B$-eigenspaces are essentially the same as $\mathfrak B$-eigenspaces, thus we shall use $\mathfrak B$-eigenspaces instead.
	
	Let $\psi\in\mathfrak B^*$ be a $\mathfrak B$-eigenvector for the coregular representation $\mathfrak B^*$ with the eigenvalue $\xi_{\psi}\in\mathfrak B^*$, namely $a\psi=\xi_{\psi}(a)\psi$ for any $a\in\mathfrak B$. It is clear that $\xi_{\psi}$ is a character of $\mathfrak B$.
	
	On one hand, by definition of coregular representation, we have $(a\psi)(1)=\psi(a\cdot 1)=\psi(a)$ for any $a\in\mathfrak B$. On the other hand, since $\psi$ is a $\mathfrak B$-eigenvector, we have
	\[
	(a\psi)(1)=\big(\xi_{\psi}(a)\psi\big)(1)=\xi_{\psi}(a)\psi(1)
	\]
	for any $a\in\mathfrak B$. Therefore $\psi(a)=\xi_{\psi}(a)\psi(1)$ for any $a\in\mathfrak B$, which means the $\mathfrak B$-eigenvector $\psi$ is proportional to the corresponding eigenvalue $\xi_\psi$. This shows that every $\mathfrak B$-eigenspace in $V$ has dimension one.
	
	For any character $\xi\in\mathfrak B^*$ and any $a,b\in\mathfrak B$, we have
	\[
	(a\xi)(b)=\xi(ab)=\xi(a)\xi(b).
	\]
	Therefore, any character $\xi\in\mathfrak B^*$ is a $\mathfrak B$-eigenvector with the eigenvalue $\xi$.
	
	Note that the characters of $\mathfrak B$ are parameterized by their kernels, that is the maximal ideals of $\mathfrak B$. Combining the facts above, we conclude that there is a bijection between $\mathfrak B$-eigenspaces and $\mathsf{Specm}(\mathfrak B)$, namely the maximal ideals of $\mathfrak B$.
	
	We then show that each generalized $\mathfrak B$-eigenspace is a cyclic $\mathfrak B$-module. We call a finite-dimensional commutative algebra $A$ \emph{local} if it has a~unique maximal ideal $\mathfrak m$. Hence it has a unique character $\zeta\in A^*$. Moreover, $\mathfrak m$ is nilpotent, see \cite[Proposition~8.6]{AM69}. Therefore, the local algebra $A$ itself as the regular representation is the generalized $A$-eigenspace corresponding to the eigenvalue $\zeta$ as $a-\zeta(a)\in \ker\zeta=\mathfrak m$. Note that every finite-dimensional commutative algebra is a direct sum of local algebras, see \cite[Theorem~8.7]{AM69}. In addition, each local summand is a generalized $A$-eigenspace with the corresponding eigenvalue uniquely determined by the summand, see the first paragraph of the proof of \cite[Theorem~8.7]{AM69}. This part now follows from the fact that $V$ is isomorphic to the regular representation of~$\mathfrak B$.
	
 It is clearly that the algebra $\mathfrak B$ is maximal commutative in $\End(V)$. The last statement follows from the first half of the proof of Lemma \ref{lem frobenius}.
\end{proof}

\subsection{Perfect integrability of Gaudin models}\label{sec main}
The following is our main theorem which asserts Gaudin models with periodic and regular quasi-periodic boundary conditions are perfectly integrable.
\begin{thm}\label{thm main}
If $\mu\in \h^*$ is regular or if $\mu =0$, then the $\mathcal A_{\bm z,\mu}$-module $\mathcal M_{\bla,\mu}$ is perfectly integrable.
\end{thm}
\begin{proof}
By Theorem \ref{thm cyclicity}, Gaudin algebra $\mathcal A_{\bm z,\mu}$ acts on $\mathcal M_{\bla,\mu}$ cyclically. Recall that $\rho_{\bla,\bm z,\mu}\colon \mathcal A_{\bm z,\mu}\allowbreak \to \End(\mathcal M_{\bla,\mu})$ and $\mathfrak A_{\bla,\bm z,\mu}=\rho_{\bla,\bm z,\mu}(\mathcal A_{\bm z,\mu})$. It remains to show that $\mathfrak A_{\bla,\bm z,\mu}$ is Frobenius.

By Lemma \ref{lem inv-form}, we can apply Lemma~\ref{lem frobenius} for the case $A=\mathcal A_{\bm z,\mu}$, $\mathfrak A=\mathfrak A_{\bla,\bm z,\mu}$, $V=\mathcal M_{\bla,\mu}$, and $(\cdot|\cdot)=\mathcal S_{\bla}(\cdot,\cdot)$. Therefore we conclude that the algebra $\mathfrak A_{\bla,\bm z,\mu}$ is a Frobenius algebra.
\end{proof}

Theorem \ref{thm main} gives the following important facts. By Theorem~\ref{thm main}, Lemma~\ref{cor intro}, and \cite[Corollary~3.3]{Ryb18}, we see that the joint eigenvectors (up to proportionality) of the Gaudin algebra in $V_{\bla}^\sing$ are in one-to-one correspondence with monodromy-free $^L\g$-opers on the projective line with regular singularities at the points $z_1,\dots,z_\ell,\infty$ and the prescribed residues at the singular points. Here $z_1,\dots,z_\ell$ are {\it arbitrary} pairwise distinct complex numbers, cf.\ \cite[Corollary~3.4]{Ryb18}. Similarly, when~$\g$ is of type~B or~C (resp.~G$_2$), one deduces from \cite[Theorem~4.5]{LMV17} (resp.\ \cite[Theorem~5.8]{LM19g2}) that there exists a bijection between joint eigenvectors (up to proportionality) of the Gaudin algebra in $V_{\bla}^\sing$ and self-dual (resp.\ self-self-dual) spaces of polynomials in a~suitable intersection of Schubert cells in Grassmannian.

\subsection[Conjecture for general $\mu\in\g^*$]{Conjecture for general $\boldsymbol{\mu\in\g^*}$}\label{sec general conj}

For an arbitrary $\mu\in\g^*\cong \g$, there exists an element $g\in G$ such that $g\mu g^{-1}$ is in the negative Borel part $\mathfrak b_-=\mathfrak n_-\oplus \h$. Thus, without loss of generality, we can assume that $\mu\in \mathfrak b_-$.

Let $\mathfrak z_{\mu}(\g)$ be the centralizer of $\mu$ in $\g$. It is known that $\mc A_{\bm z,\mu}$ commutes with the diagonal action of $\mathfrak z_{\mu}(\g)$, see \cite[Proposition~4]{Ryb06}.

Let $V_{\bla}$ be as before. Define $\mc M_{\bla,\mu}$ as a subspace of $V_{\bla}$ by
\[
\mc M_{\bla,\mu}:=\{v\in V_{\bla}\,|\, x v=0, \text{ for all }x\in \mathfrak z_{\mu}(\g)\cap \mathfrak n_+\}.
\]
Then $\mc A_{\bm z,\mu}$ acts on $\mc M_{\bla,\mu}$.
\begin{conj}\label{conj frob-ger}
The $\mc A_{\bm z,\mu}$-module $\mc M_{\bla,\mu}$ is perfectly integrable.
\end{conj}

\subsection*{Acknowledgements}

The author is grateful to E.~Mukhin and V.~Tarasov for interesting discussions and helpful suggestions. The author also thanks
the referees for their comments and suggestions that substantially improved the first version of this paper. This work was partially supported by a grant from the Simons Foundation~\#353831.

\pdfbookmark[1]{References}{ref}
\LastPageEnding

\end{document}